\documentclass[envcountsame,envcountsect,12pt,undated,nolinenumbers]{lnthi}

\usepackage{amsmath,graphics}

\def\guidesort{Guide\-sort}
\def\DL{D\Tsub L}

\setpagewiselinenumbers

\title{Guidesort: Simpler Optimal Deterministic
 Sorting
 for the Parallel Disk Model}

\author{Torben Hagerup}

\institute{\Tinfuna[5]\\
  \email{hagerup@informatik.uni-augsburg.de}}

\begin{document}
  \overfullrule=5pt

\maketitle{}

\begin{abstract}
A new algorithm, \guidesort, for sorting in the uniprocessor
variant of the parallel
disk model (PDM) of Vitter and Shriver
is presented.
The algorithm is deterministic and executes
a number of (parallel) I/O operations that comes
within a constant factor $C$ of the optimum.
The algorithm and its analysis are simpler
than those proposed in previous work, and the achievable
constant factor $C$ of essentially~3 appears to be smaller than
for all other known deterministic algorithms,
at least for plausible parameter values.

\vspace{8pt plus 2pt minus 1pt}

\guidesort.
\end{abstract}

\pagestyle{plain}
\thispagestyle{plain}

\section{Introduction}
\label{sec:intro}

Sorting is an important problem.
In the seventies
Knuth quoted an estimate that over
25\% of computers' running time is spent on
sorting~\cite{Knu73}.
Frequently the data sets to be sorted are so large that
they do not fit in internal memory and must be held
on external storage, often one or more magnetic disks
or similar devices.
In this setting of massive data,
sorting acquires even more importance
because a high number of algorithms that use external
storage efficiently do so by reducing other problems
to sorting, so that the time to sort can almost be
viewed as playing the role that linear time has
in RAM computation.

Accesses to magnetic disks are much slower than
CPU operations and accesses
to internal memory.
A natural and increasingly popular way to sort faster
is to use many disks in parallel.
Magnetic disks have high
latencies, so efficiency dictates that an access
to a magnetic disk must be used to transfer not one
data item, but a whole block of
many data items.
The parallel disk model
of Vitter and Shriver~\cite{ShrV94} tries to capture
these characteristics of disk systems and has been
used for much of the extensive research on sorting with
several disks.

\subsection{Model and Problem Statement}

An instance of the uniprocessor variant of the
\emph{parallel disk model} or \emph{PDM}
of Vitter and Shriver~\cite{ShrV94} is specified
via three positive integers, $M$, $B$ and $D$
with $2 B\le M$.
It features a machine comprising
an \emph{internal memory} of
$M$ cells and $D$ \emph{disks},
each with an infinite number of cells.
Every disk is
linearly ordered and
partitioned into \emph{block frames} of
$B$ consecutive
cells, each of which can accommodate
a \emph{block} of $B$ data items.
In slight deviation from the original
definition of the PDM, we will assume that the
internal memory is also linearly ordered
and partitioned into block frames of
$B$ consecutive cells and---therefore---that
$M/B$ is an integer.
An \emph{I/O operation} or, for short, an
\emph{I/O} can copy the blocks stored in
$D$ pairwise distinct block frames in the
internal memory to block frames
on $D$ pairwise distinct disks
(an \emph{output operation}) or vice versa
(an \emph{input operation}).
Thus the disks are assumed to be
synchronized.
The description in~\cite{ShrV94} does not
indicate explicitly whether it is also possible
to transfer fewer than $D$ blocks between
pairwise distinct block frames in the
internal memory and block frames
on as many pairwise distinct disks;
here we will assume this to be the case.
Arbitrary computation
(\emph{internal computation})
can take place on data
stored in the internal memory, whereas
items stored on disks can participate in computation
only after being input to the internal memory.
An algorithm is judged primarily by the number
of I/O operations that it executes, computation
in the internal memory usually being considered free.

So that operations in the internal memory and
I/O operations can specify their arguments, we
assume that the disks, the cells in the internal memory
and the block frames both in the internal
memory and on each disk are numbered
consecutively, starting at~0 (say).
A number of conventions make the PDM convenient to argue
about, but less precise.
First, the meaning of a ``cell'' depends on the
problem under consideration.
In the context of sorting, as relevant here,
a cell is the amount of memory needed
to store one of the items to be sorted
or a comparable object
(e.g., for sorting it is generally assumed that
one can ensure at no cost that the keys of the
input items are pairwise distinct by appending
to each its position in the input).
Second, the only space accounted
for is that taken up by ``data items'', not that
needed to realize the control structure
of an algorithm under execution.
E.g., algorithms like those discussed in the
following may want to manipulate such data
structures as a recursion stack and
various vectors indexed by disk numbers.
Space for such bookkeeping information is
assumed implicitly to be available
whenever needed.

It is customary to restrict the parallel
disk model by imposing an additional condition that says,
informally, that $M$ is sufficiently large
relative to $B$ and $D$.
The condition varies from description to
description, however, and seems to reflect the
requirements of particular algorithms more than
any fundamental
deliberation concerning the model.
E.g., \cite{ShrV94} requires that
$D\le\Tfloor{{M/B}}$, \cite{BarGV97}
that $M\ge 2 D B$, and \cite{NodV95} that
$D B\le\Tfloor{{{(M-M^\beta)}/2}}$ for
some fixed $\beta<1$.
Since no more than $M/B$ disks can take part
in a (parallel) I/O operation, the condition
$D\le {M/B}$ seems somewhat more canonical than
the others, and we will adopt it here
(alternatively, $D$ can be interpreted as the
minimum of $M/B$ and a true number of disks).
It cannot be a priori excluded, however, that
more than $M/B$ disks can be put to good
use in a PDM algorithm
(cf.\ the RAMBO model of
Fredman and Saks~\cite{FreS89}).

A sequence of items is said to be stored in
the \emph{striped} format
if its blocks are distributed
over the disks in a round-robin fashion.
More precisely, a sequence $(x_0,\ldots,x_{N-1})$
is stored in the striped
format if there are (known)
nonnegative integers $a_0,\ldots,a_{D-1}$ such
that the following
holds for $i=0,\ldots,N-1$:
If $\Tfloor{{i/B}}=q D+r$ for integers $q$ and $r$
with $q\ge 0$ and $0\le r<D$, then $x_i$ is
the item numbered $i\bmod B$
in the block frame numbered $a_r+q$
on the disk numbered~$r$.
Every sequence for which nothing else
is stated explicitly
in the following is assumed
to be stored in the striped format.

The problem of sorting in the PDM is
defined as follows:
Given as input
a sequence $(x_0,\ldots,x_{N-1})$,
stored in the striped format,
of items with a partial order defined by keys
drawn from a totally ordered universe,
output a sequence of the form
$(x_{\sigma(0)},\ldots,x_{\sigma(N-1)})$,
again stored in the striped format,
where $\sigma$ is a bijection from
$\{0,\ldots,N-1\}$ to $\{0,\ldots,N-1\}$
with
$x_{\sigma(0)}\le\cdots\le x_{\sigma(N-1)}$.

\subsection{The Challenge}
\label{subsec:challenge}

In the remainder of the paper, consider
the problem of sorting a sequence of
$N>M$ items
and take $n=\Tceil{N/B}$
(the number of block frames occupied by the input)
and $m={M/B}$
(the number of block frames in the internal memory).
Without loss of generality we will assume
that the $N$ items to be sorted have
pairwise distinct keys.
The last of the $n$ blocks of input may contain
a segment of data beyond the $N$ items to be sorted.
Such a ``foreign''
segment should be treated as a sorted
sequence of dummy items larger than all real items,
so that the segment will not be modified by the sorting.

In the sequential case, i.e., for $D=1$,
we can sort within
$\Tvn{Sort}_{M,B}(N)=2 n(1+\Tceil{\log_m({N/M})})=
2 n\Tceil{\log_m\! n}$ I/Os
by first forming $\Tceil{N/M}$ sorted runs
of at most $M$ items each
and then merging the runs in an $m$-ary
tree of height $\Tceil{\log_m({N/M})}$
(a more practical algorithm uses $(m-1)$-way
merging instead of $m$-way merging
at a slight loss of theoretical efficiency).
A method for obtaining a
nearly matching lower bound of
${{\Tvn{Sort}_{M,B}(N)}/(1+F(M,B,N))}$,
where
\[
F(M,B,N)=O\left(
\frac{1}{\ln m}+\frac{1}{\log_m\! n}
+\frac{\log_m\! n}{B}
\right),
\]
was indicated
in \cite{AggV88,HutSV05}
(for models of computation no weaker than the PDM).
If $m$, $\log_m\!n$
and $B/{\log_m\! n}$
all tend to infinity
simultaneously,
the ratio between the upper bound
$\Tvn{Sort}_{M,B}(N)$ and the lower bound above
tends to~1, which is why the leading factor of
the complexity of sorting in the PDM
with $D=1$ can be claimed to be known.
Within the PDM,
a machine with a single disk can obviously
simulate one with $D$ disks with a slowdown
of~$D$, so a lower bound of
$({1/D}){{\Tvn{Sort}_{M,B}(N)}/(1+F(M,B,N))}$
holds for sorting in the (uniprocessor)
PDM with $D$ disks.
Our goal here is to prove a corresponding
upper bound of
$({C/D})\Tvn{Sort}_{M,B}(N)(1+F'(M,B,D,N))$,
where $F'$ is similar to~$F$
and $C$ is a small constant.
There is also a lower bound of
$\Omega(({1/D})\Tvn{Sort}_{M,B}(N))$ I/Os~\cite{AggV88}.
For this reason, algorithms that sort with
$O(({1/D})\Tvn{Sort}_{M,B}(N))$
I/Os are often said to be \emph{optimal}.

A PDM machine with $D$ disks can
simulate one with a single disk but a larger
block size of $B'=B D$ without slowdown
by operating the
disks in lock-step, i.e., groups of corresponding
block frames, one from each disk, are formed once and
for all, and every I/O operation inputs from or outputs to
the block frames in a group.
Applying this to the sequential merging
algorithm discussed above and assuming
that $m'=m/D$ is an integer and at least~2
yields a sorting algorithm that uses at most
$2\Tceil{n/D}\Tceil{\log_{m'}({n/D})}$ I/Os.
Call this algorithm sorting
by \emph{naive $m'$-way striping}.
If $D$ is sufficiently small relative to~$m$, there
is no significant difference between logarithms
to base $m$ and logarithms to base $m'$.
In particular, if
$D=O(m^{1-\epsilon})$
for some fixed $\epsilon>0$, the bound of
$2\Tceil{n/D}\Tceil{\log_{m'}({n/D})}$ I/Os
is within a constant factor of the
second lower bound.
As $D$ approaches $m$, however, the problem
becomes increasingly difficult.
An attempt to parallelize the sequential
merging algorithm simply by inputting the $D$ next
blocks in parallel meets with the difficulty
that the $D$ next blocks may not be known and,
even if they are, may not be stored on $D$
distinct disks.
This could be called the problem of
\emph{read contention}.

\subsection{The New Result}

We present a new algorithm, \guidesort,
for sorting in the PDM.
The algorithm is deterministic,
simple and
easy to analyze and to implement.
If the factor of $1+F'(M,B,D,N)$ of the
previous subsection is
ignored, the number of I/Os executed by
\guidesort\ is $({C/D})\Tvn{Sort}_{M,B}(N)$,
where $C$ is approximately~3
for typical values of $M$, $B$ and~$D$---for
brevity, the \emph{constant factor}
of \guidesort\ is~3.
As $D$ becomes large relative to
$m$ and~$B$, $C$ grows
to a maximum of around~9.

\guidesort\ works by computing a \emph{guide}
that can be used to redistribute blocks to
disks in such a way that the read-contention
problem disappears.
More details are provided in
Section~\ref{sec:guidesort}.

\subsection{Previous Work}

As befits a fundamental problem of great
practical importance,
a high number of algorithms for sorting in
the PDM has been proposed.
Many were qualified as simple.
We are not aware, however, of any previously
published algorithm that can be proved
efficient using simple arguments.
Some of the algorithms work via repeated
merging~\cite{AggP94,BarGV97,DemS03,HutSV05,KunR11,NodV95,RajS08}
and hence bottom-up
(this also applies to \guidesort),
while others use the
approximately inverse process of repeated splitting
at chosen partitioning
elements~\cite{HutSV05,NodV93,HutV06,ShrV94}
and therefore operate in a top-down fashion.

Most published algorithms resort to randomization
to cope with the read-contention problem.
Some of
them~\cite{BarGV97,HutSV05,HutV06}
are optimal, as concerns the expected number
of I/Os, and
achieve constant factors close to~1
if $D$ is sufficiently small relative to~$m$
(this is also when sorting by naive
striping comes into its own),
but as $D$ approaches $m$ their performance degrades.

Explicit constant factors were not indicated
for any of the optimal deterministic algorithms
published to date.
The scheme of Aggarwal and Plaxton~\cite{AggP94},
based on Sharesort and bottom-up, but with some
elements of top-down, is applicable to a
variety of computational models.
Its constant factor seems difficult to determine,
but the generality of the approach lets one
expect it to be quite large.
Balance Sort by Nodine and
Vitter~\cite{NodV93} is a deterministic
top-down algorithm that depends on subroutines
for complicated tasks such as load balancing,
matching and derandomization.
Again, the constant factor appears to be large.

Greed Sort, also due to Nodine and
Vitter~\cite{NodV95}, is perhaps closest in
spirit to the \guidesort\ algorithm
presented here.
It is deterministic, based on
repeated merging,
and sufficiently simple that estimating
its constant factor seems feasible.
In order to merge, Greed Sort first
carries out an approximate $R$-way merge,
for a certain $R$,
that brings each item to a
position within some distance $L$ of
its rank in the sorted combined sequence,
and then finishes by using Leighton's
Columnsort~\cite{Lei85} to sort locally
within overlapping segments of $2 L$ items each.

As indicated in~\cite{NodV95}, the number
of I/Os executed by the approximate merge
of Greed Sort is
between 3 and~5 times the number of blocks
involved.
Subsequently each block participates in
two applications of Columnsort.
Columnsort, in turn, consists of 8 steps,
each of which reads and writes every block
at least once.
Thus each recursive level needs at least
$3+2\cdot 8\cdot 2=35$ I/Os per block.
In addition, $R$ is chosen approximately
as $\sqrt{m}$, so that sorting based on $R$-way
merging has about twice as many recursive
levels as sorting by $m$-way merging.
Since each recursive level in sequential
sorting reads and writes each block just
once ($2$ I/Os), this calculation indicates
a constant factor
for Greed Sort
of at least~35.
In fact, a more detailed study of~\cite{NodV95}
reveals the estimate of~35 to be optimistic,
especially if $D$ is not much smaller than~$m$.

The work described here was borne out of the
author's desire to have an optimal sorting algorithm
for the PDM simple enough to serve as the basis of
a homework problem for students.
Whereas this seemed out of the question for all
previously published algorithms, the plan was
carried out successfully for \guidesort.

\section{\guidesort}
\label{sec:guidesort}

\guidesort\
still sorts recursively, with
the base case given by internal sorting
of at most $M$ items and each nonterminal
call executing a multiway merge of
recursively sorted sequences.
Only the merge is done in a novel way.

\subsection{The Main Idea}

Call each sorted input sequence
of a merge a \emph{run} and
consider a multiway merge of
runs partitioned into blocks.
Assume first that the blocks are input one by one.
As long as there is enough internal memory to store
one block from each run and
to buffer the output, each input block
must be read only once.

Define the \emph{leader} of a block of items to be its
smallest item.
It is natural to read the next block from a run
as soon as the last item in the previous
block from that run has been consumed,
i.e., moved to the output buffer.
If we knew the leader $x$ of the next block
in the run, however, we could postpone reading
that block until just before $x$ is to be consumed.
This shows that if the blocks of all
runs are input in the order given by the sorted order
of their leaders, it is still the case that each
input block must be read only once.
We can discover the
appropriate \emph{canonical order}
by forming the \emph{sample} of each
run as the sorted sequence of its leaders
and merging the samples, which creates
the \emph{canonical sequence}.
Since the samples are far smaller than the full
runs, the cost of merging them
is usually negligible.
The computation of the canonical sequence
was also considered by
Hutchinson, Sanders and
Vitter~\cite{HutSV05}, who used the
term ``trigger'' for what
we call a leader.

Suppose now that we manage to input the blocks in
the canonical order in batches of $\overline{D}\le D$
consecutive blocks each, where each batch is input
in a single I/O operation.
If we provide an additional buffer of $\overline{D}$ block
frames to hold the latest input batch, it will still
be the case that each input block must be read only once.
In order for it to be possible to consume the
input in batches, the blocks in each
batch must be stored on different disks.
In addition, to make it possible to produce
each run $\overline{D}$ blocks at a time, we
must ensure that if the run is partitioned
into subsequences of $\overline{D}$ consecutive blocks each,
the blocks within each subsequence are stored
on different disks.
Viewed more abstractly, we are faced with the
problem of coloring
the vertices of a graph $G$ defined by the
canonical sequence
with at most $D$ colors
in such a way that no two
adjacent vertices receive the same color.
The vertices of $G$ correspond to the leaders
or the blocks in all runs, each edge
joins two blocks that must be stored on
distinct disks, and the $D$ colors
correspond to the $D$ disks.
The degree of $G$ is bounded by $2(\overline{D}-1)$,
since each block shares its batch with
$\overline{D}-1$ other blocks and its
subsequence with $\overline{D}-1$ other blocks.
If we choose $\overline{D}=D$, coloring $G$ may be difficult
or impossible.
With $\overline{D}\le\Tceil{D/2}$, however,
$2(\overline{D}-1)<D$, so that the
obvious greedy algorithm
can color the canonical sequence in
a single pass.
Call the resulting sequence of colored
leaders the \emph{guide}.
While keeping the full guide, we also transfer
the colors of leaders found to the
original samples.
This can be done with a recursive splitting
that reverses the steps of the
merge that created the canonical sequence.

To carry out the overall merge, we first process each
run, redistributing its blocks to the disks
specified in its colored sample.
This can be accomplished in a single pass over the
run synchronized with a pass over
its colored sample.
Subsequently the overall merge of the
runs can proceed under the control
of the guide, $\overline{D}$ blocks at a time,
and this process can also create
the corresponding sample for the
merge at the next higher recursive level at
essentially no additional cost.

To summarize,
Guidesort merges using the following steps,
illustrated in Fig.~\ref{fig:1}:
\begin{enumerate}
\item
Merge the samples to obtain the canonical sequence.
\item
Color the canonical sequence to obtain the guide.
\item
Split the guide into colored samples.
\item
Redistribute the blocks of each
run to new disks as indicated by its colored sample.
\item
Actually merge the
runs, always using the guide to know where to read
the next batch of input blocks,
and generate the corresponding sample.
\end{enumerate}

\begin{figure}
\begin{center}
\includegraphics{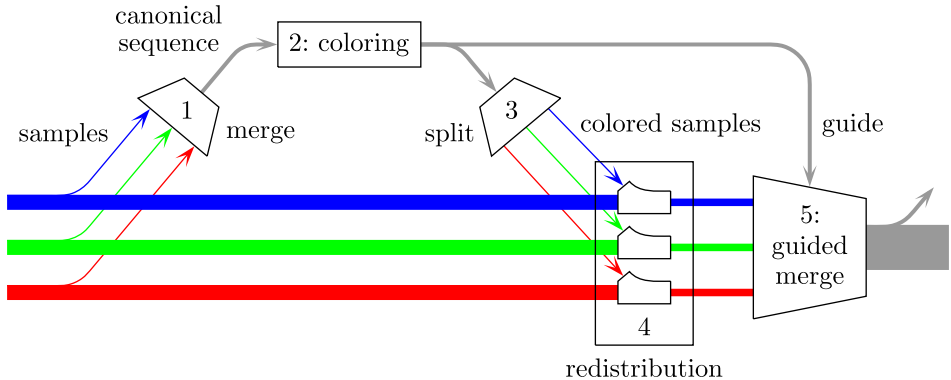}
\end{center}
\caption{The five steps of Guidesort's merge.
Different colors or gray tones
in the figure denote different runs
and should not be confused with the colors used by
Guidesort, which represent different disks.}
\label{fig:1}
\end{figure}

A back-of-the-envelope estimate of the number of
I/Os executed by
Guidesort proceeds as follows:
Steps 1--3 operate only on leaders and therefore
have no significant I/O cost.
Step~4 inputs all items and outputs them at
half speed (because we can choose $\overline{D}\approx{D/2}$).
Conversely, Step~5 inputs the items at half
speed and outputs them at full speed.
Altogether, each merge is about three times as expensive
as would be
inputting and outputting the
relevant blocks once $D$ at a time.
As a consequence, the complete sorting uses
approximately $({3/D})\Tvn{Sort}_{M,B}(N)$ I/Os.

While the back-of-the-envelope estimate essentially
leads to the correct result, a more precise
argument must account for the cost of Steps
1--3, which is not always negligible,
and must also work without a number of
assumptions that were made implicitly above.
This is the topic of the next two subsections.

\subsection{A
Realistic Simple Special Case}

Define the \emph{block length} of a stored sequence to
be the number of block frames that it occupies
(fully or in part).
Let us be more specific about the recursive
sorting algorithm, which is parameterized by
an integer $r$ with $2\le r\le m$.
When asked to sort a sequence of block length $p$,
the algorithm first computes $k=\min\{\Tceil{{p/m}},r\}$.
If $k=1$, the sorting takes place in the
internal memory without recursive calls.
Otherwise the input sequence is split into
$k$ subsequences, each of block length $\Tfloor{{p/k}}$
or $\Tceil{{p/k}}$, and the $k$ subsequences
are sorted recursively and subsequently merged.
The main point is that although the algorithm
``nominally'' uses $r$-way merging, a call
that issues terminal calls issues only as
many as necessary.
Without increasing the depth of the recursion,
this ensures that each terminal call deals
with at least $m/2$ blocks and therefore
that the number of terminal calls and the
number of calls altogether are $O({n/m})$.
This is useful for dealing with rounding issues.
E.g., if each merge executed as part of the
overall sorting is associated with a quantity of
the form $\Tceil{t}$, where $t$ is some expression,
we can upper-bound the sum of $\Tceil{t}$ over
all merges by summing $t$ over all merges
and adding $O({n/m})$.

Our analysis frequently sums the number of I/Os
of a particular kind over all merges carried out
as part of the overall sorting.
We shall use the term ``accumulated'' to denote this
situation, i.e., ``accumulated'' means
``summed over all merges''.
Let $z$ be
the accumulated total block length of the
runs of all merges.
Exactly as in the sequential sorting algorithm,
$z\le n\Tceil{\log_r\!n}$.
Similarly, let $y$ be the accumulated block
length of all samples produced by the algorithm.
Since each leader in a sample
``represents'' $B$ items,
an application of the counting method set out
at the end of the previous paragraph shows
that $y={z/B}+O({n/m})$.

Let $D_1=\min\{D,\Tfloor{m/2}\}
\ge\max\{\Tfloor{D/2},1\} \ge{D/3}$.
In the following, $D_2$, $D_4$, $D_5$
and $\DL$ are integers with
$1\le D_2,D_4,D_5,\DL\le D$
to be chosen later.

\textbf{Step~1} merges at most $r$ samples.
If the block length of some sample is $<m$,
and therefore the block lengths of all samples
are $\le m$, the samples are first partitioned
into \emph{bundles} in such a way that the
total block length of the runs in each bundle,
except possibly the last bundle, is at least
$m/2$ and at most~$m$.
The blocks in each bundle are then input,
$D$ at a time, sorted internally, and output
as one sorted sequence, again $D$ blocks at a time,
at an accumulated cost of ${{2 y}/D}+O({n/m})$ I/Os.
Assume that the number of sorted runs
(original samples or sorted bundles)
at this point is $r'$.
The $r'$ runs are merged in a binary merge tree
of height $\Tceil{\log_2\! r'}$ in which each
binary merge
is carried out with $D_1$
disks that operate in lock-step to simulate
a single disk with a block size of $B'=D_1 B$.

Observe that
${n/m}\le{n/D}
=O(({z/D})({1/{\log_m\! n}}))$
(a convexity argument shows that
we even have
${n/D}\le({z/D})({1/{\log_m\! n}})$).
When $r'$ runs are merged in a binary merge tree,
each run, except possibly the last one,
contains at least $M/2$ leaders
that ``represent'' at least ${{M B}/2}$
input items.
With respect to a block size of $B'$, the
accumulated block length of all runs
input to nontrivial binary merge
trees---those with $r'\ge 2$ leaves---is therefore
$O({z/{(D_1 B)}})=O({z/{(D B)}})$,
and the merges can be carried out with
$O({{z\Tceil{\log_2\! r}}/{(D B)}})
=O({{(z\ln m)}/{(D B)}})$ I/Os.
Altogether, the number of I/Os executed by Step~1 is
$O({{(z\ln m)}/{(D B)}}+{n/m})$.
This bound also covers
\textbf{Step~3},
which reverses the merging to transfer the
information attached to leaders in Step~2
from the guide to the samples.
Steps 1 and~3
need $2 D_1$ block frames of internal
memory, which are available since $D_1\le{m/2}$.

\textbf{Step~2} colors the canonical sequence.
Say that a color is used \emph{recently} in a
sequence of $k$ colored objects if it is the
color of one of the $\min\{\overline{D}-1,k\}$
last objects in the sequence.
The leaders are processed in the order in which
they occur in the canonical sequence, and each
leader $x$ is given a color
in $\{0,\ldots,D-1\}$ that is used recently neither
in the complete sorted sequence of colored
leaders nor in its (not
necessarily contiguous) subsequence of leaders
drawn from the same sample as~$x$.
To accomplish its task, the algorithm maintains
at most $r+1$ \emph{history}
sequences of
the chronologically ordered recently used colors
overall and within each sample.
To process a leader, the algorithm identifies
and assigns an appropriate color
to the leader and updates two
history sequences accordingly,
all of which is straightforward.
In addition, the leaders of each color are numbered
consecutively in the order in which they are colored
and the number of each leader, called its \emph{index},
is attached to the leader.
We will assume that the history
sequences with their at most
$(r+1)(\overline{D}-1)$ color values can be stored
in $q=\Tceil{{(r+1)(\overline{D}-1)}/B}$ block frames.
Using an additional input/output buffer of
$D_2$ block frames,
we can execute Step~2 with an
accumulated number of I/Os of
${{2 y}/{D_2}}+O({n/m})={{2 z}/{(D_2 B)}}+O({n/m})$.
We must ensure that $q+D_2\le m$.

\textbf{Step~4} redistributes blocks to new disks
one run at a time.
When processing a run of $\ell$ blocks,
the algorithm inputs the $\ell$ blocks and,
interleaved, a sample of $\ell$ colored leaders and,
for $i=1,\ldots,\ell$, stores the $i$th block
on the disk corresponding to the color
of the $i$th leader and in a block frame
whose number is the index of the $i$th leader plus
an offset chosen for that disk.
Thus the blocks on each disk are stored compactly,
but they are not necessarily written in the order
of increasing frame numbers.
With a ``primary'' input buffer of $D_4$
block frames for the input blocks, a ``secondary''
input buffer of $\DL$ block frames
for the leaders and an output buffer of
$\overline{D}$ block frames for the
redistributed blocks,
the accumulated number of I/Os spent in
Step~4 is
${{z/{D_4}}}+{{z/{\overline{D}}}}+
{y/{\DL}}+O({n/m})=
z({1/{D_4}}+{1/{\overline{D}}}+{1/{(\DL B)}})+O({n/m})$.
We must ensure that
$D_4+\overline{D}+\DL\le m$.

\textbf{Step~5} actually merges the runs with
the aid of
one block frame of input buffer for
each run, a total of at most $r$ block frames,
an input buffer of $\overline{D}$ block frames for
the latest batch
($\overline{D}-1$ block frames actually suffice),
two buffers of $\DL$ block frames each for
the guide, which is input, and the sample for the
next recursive level, which is output, and finally
a buffer of $D_5$ block frames
for the primary output, the final
outcome of the merge.
The accumulated number of I/Os is
${{z/{\overline{D}}}}+{{z/{D_5}}}
+2{{y/{\DL}}}+O({n/m})=
z({1/{\overline{D}}}+{1/{D_5}}+{2/{(\DL B)}})
+O({n/m})$,
and we must ensure that
$r+\overline{D}+D_5+2\DL\le m$.

Collecting the contributions identified above
and not forgetting the $O({n/D})$ I/Os
consumed by terminal calls of the recursive
sorting algorithm,
we arrive at a total number of I/Os
for the complete sorting of

\begin{equation}
\label{eq:IO}
z\left(\frac{1}{D_4}+\frac{1}{D_5}+\frac{2}{\overline{D}}
 +\frac{1}{B}\left(\frac{2}{D_2}+\frac{3}{\DL}\right)\right)
 +O\left(\frac{z\ln m}{D B}+\frac{n}{D}\right).
\end{equation}

Assume that $m\ge 6 D$, $B\ge D$ and
$\ln m=o(B)$,
conditions
that are not unlikely to be met in a
practical setting
(of course, it is not really clear what
$\ln m=o(B)$ is supposed to mean in a
practical setting).
Then we can satisfy all requirements identified
in the discussion of Steps 2, 4 and~5 by taking
$D_2=D_4=D_5=\DL=D$,
$\overline{D}=\Tceil{D/2}$ and
$r=m-4 D\ge 2$
(in particular, $q\le r+1$),
and the total number of I/Os becomes at most
$({z/D})(6+O({1/{\log_m\! n}})+o(1))$.
Since $r\ge{m/3}$ and therefore
$\ln m=(1+O({1/{\ln m}}))\ln r$,
the number of I/Os
can also be bounded by
$({3/D})\Tvn{Sort}_{M,B}(N)(1+O({1/{\log_m\! n}})+o(1))$,
where $n$ and $m$ are assumed to tend to infinity.

\subsection{The General Case: An Order-of-Magnitude Bound}

In order to deal with situations in which
$B$ is smaller than $D$
or $\ln m$ is not small relative to $B$,
we generalize the algorithm by introducing an
additional parameter, $s$, which must be a
positive divisor of $\overline{D}$.
The algorithm described so far corresponds to
the special case $s=1$.

Whereas until now a leader was the smallest
item within a block, we redefine it to
be the smallest item within a \emph{segment}
of $s$ consecutive blocks in some run.
Thus each run is partitioned into segments
of $s B$ items each, except that the last
segment may be smaller, and each segment
contributes only a single leader to the
sample of the run.
As a consequence, the accumulated block
length of all samples now is
$y={z/{(s B)}}+O({n/m})$.
We stipulate that the $s$ blocks that form
a segment are colored using consecutive
colors that begin at a multiple of~$s$.
An intuitive view of this is that the
task now is to color segments with one
of ${{2\overline{D}}/s}-1$ colors so as to avoid the at most
${{\overline{D}}/s}-1$ most recent
colors both overall and within the same run.
This has the beneficial effect of reducing
the state information that must be kept
by the greedy coloring
in Step~2 from $(r+1)(\overline{D}-1)$ colors to
at most ${{(r+1)(\overline{D}-1)}/s}$ colors,
which can be stored in
$q=\Tceil{{{(r+1)(\overline{D}-1)}/{(s B)}}}$ blocks.
On the other hand, in order for the leaders
to fulfil their function, Step~5 must
be changed to input runs whole segments
at a time, which means that an input buffer of
$s$ block frames
rather than of a single block frame
must be provided for each run.
The requirement for Step~5 accordingly
becomes $r s+\overline{D}+D_5+2\DL\le m$
for the revised sorting algorithm,
while the total number of I/Os
generalizes from~(1) to

\begin{equation}
\label{eq:IOs}
z\left(\frac{1}{D_4}+\frac{1}{D_5}+\frac{2}{\overline{D}}
 +\frac{1}{s B}\left(\frac{2}{D_2}+\frac{3}{\DL}\right)\right)
 +O\left(\frac{z\ln m}{s D B}+\frac{n}{D}\right).
\end{equation}

For $m\le 3$ and therefore $D\le 3$,
sequential sorting obviously works
within $({3/D})\Tvn{Sort}_{M,B}(N)$ I/Os.
If $m\ge 4$ and $D\le\sqrt{m}$,
sorting by naive $\Tfloor{\sqrt{m}}$-way
striping uses
$({2/D})\Tvn{Sort}_{M,B}(N)(1+O({1/{\sqrt{m}}}))+O({n/D})$
I/Os.
Both of these bounds are better than the bound
claimed in Theorem~\ref{thm:main}
in the next subsection.
We will therefore assume in the following
that $D\ge\sqrt{m}$.

If the only goal is to prove a bound of
$O(({1/D})\Tvn{Sort}_{M,B}(N))$, i.e., if
constant factors are not considered significant, we can
assume without loss of generality that
$\sqrt{D}\ge 12$ and
argue as follows:
Take $s=\Tfloor{{{\sqrt{m}}/2}}
\ge\Tfloor{{{\sqrt{D}}/2}}\ge 6$ and
choose $\overline{D}$ as the largest multiple
of $s$ bounded by $D/2$---which is $\Omega(D)$
since $s\le{{\sqrt{m}}/2}\le{D/2}$.
Moreover, take $r=\Tfloor{s/2}-1\ge 2$
and $D_2=D_4=D_5=\DL=\Tfloor{D/8}\le{m/8}$.
It is now easy to see that the requirements of
Steps 2, 4 and~5 are satisfied.
In particular,
$q=\Tceil{{{(r+1)(\overline{D}-1)}/{(s B)}}}
\le\Tceil{{{\overline{D}}/{(2 B)}}}\le\Tceil{{m/4}}$
and $r s\le ({1/2})({{\sqrt{m}}/2})^2={m/8}$.
The number of I/Os executed by the algorithm
is $O({z/D})$.
Moreover, $r=\Omega(\sqrt{m})$ and therefore
$\ln r=\Omega(\ln m)$, so the number of I/Os is indeed
$O(({1/D})\Tvn{Sort}_{M,B}(N))$.
If we want to prove a more precise bound,
we must choose the parameters more carefully.

\subsection{The General Case: Good Constant Factors}

In order to obtain the best result, we change the
algorithm slightly:
In Step~4, instead of having separate primary input and
output buffers of $D_4$ and $\overline{D}$ block
frames, respectively, we use a single buffer of $D_4$
block frames for both input and output.
This is trivial, but requires us to ensure that
$\overline{D}\mid D_4$ (i.e., $\overline{D}$ divides $D_4$).

Besides $r$, $s$, $\overline{D}$, $D_2$, $D_4$, $D_5$
and $\DL$ being positive integers, the conditions
that the parameters must satisfy for the
final algorithm are the following:
\begin{align*}
\overline{D}&\le\Tceil{D/2}\\
D_2,D_4,D_5,\DL&\le D\\
{{(r+1)(\overline{D}-1)}/{(s B)}}+D_2&\le m\\
D_4+\DL&\le m\tag{3}\\
r s+\overline{D}+D_5+2\DL&\le m\\
2\le r&\le m\\
s\mid\overline{D}\mbox{ and }\overline{D}&\mid D_4.
\end{align*}

Our task at this point is to (approximately) minimize the
I/O bound of (\ref{eq:IOs}) subject to the
constraints~(3).
As an aid in dealing with the divisibility
requirements of the last constraint,
we first prove a simple technical lemma that,
informally, says that, given two positive
integers $a$ and $b$, we can make
the smaller divide the larger
without changing their sorted
order by lowering each by less than half
and $b$ by less than its square root.
Since we will actually use the lemma only with
$a\le\sqrt{b}$, it is more general than
what is needed here.
Let $\TbbbN=\{1,2,\ldots\}$.

\begin{lemma}
\label{lem:divisibility}
There is a function $f:\TbbbN^2\to\TbbbN^2$
that can be evaluated in constant time and
has the following properties:
For all $a,b\in\TbbbN$, if $(a',b')=f(a,b)$, then
\begin{itemize}
\item[$\bullet$]
$a'\mid b'$ or $b'\mid a'$;
\item[$\bullet$]
${a/2}<a'\le a$;
\item[$\bullet$]
$\max\{{b/2},b-\sqrt{b}\}<b'\le b$;
\item[$\bullet$]
$(a-b)(a'-b')\ge 0$.
\end{itemize}
\end{lemma}

\begin{proof}
If $a=1$, take $a'=a$ and $b'=b$.
Otherwise, if $a\ge b$, take $a'=\Tfloor{a/b}b>{a/2}$ and $b'=b$.
If $a=b-1\ge 2$, take $a'=b'=a>\max\{b/2,b-\sqrt{b}\}$.
From now on assume that $2\le a\le b-2$.
Since then $b\ge 4$ and $\sqrt{b}\le{b/2}$, we need no longer
verify the condition
$b'>{b/2}$ explicitly.
Consider two cases:

If $a\le\sqrt{b}+1$, take $a'=a$
and $b'=\Tfloor{b/a}a\le b$ and
observe that if $a<\sqrt{b}+1$, then
$b'=\Tfloor{b/a}a\ge
b-a+1>b-\sqrt{b}$, whereas if $a=\sqrt{b}+1$, then
$b'=\Tfloor{b/a}a\ge\Tfloor{\sqrt{b}-1}(\sqrt{b}+1)
=b-1>b-\sqrt{b}$.

If $\sqrt{b}+1<a\le b-2$,
take
$a'=\Tfloor{b/{\Tceil{b/a}}}\le a$ and
observe that if $a\ge{b/2}$, then
$a'=\Tfloor{b/2}\ge\Tfloor{{{(a+2)}/2}}>{a/2}$,
whereas if $a<{b/2}$, then
we can successively conclude that
$a\ge 4$, $b\ge 9$ and $a\ge 5$
and then that
$a'\ge
\Tfloor{b/{({b/a}+1)}}=\Tfloor{{a b}/{(b+a)}}
\ge\Tfloor{{{2 a}/3}}>{a/2}$.
Also take
\[
b'=\Tceil{b/a}a'
=\Tceil{b/a}\left\lfloor{\frac{b}{\Tceil{b/a}}}\right\rfloor
\ge
b-\Tceil{b/a}+1
\ge b-\Tceil{\sqrt{b}}+1
>b-\sqrt{b}
\]
and note that $b'\le b$ is obvious from the third expression
in the previous line.
\end{proof}

\noindent
The stage is set for our main result:

\begin{theorem}
\label{thm:main}
For all positive integers $M$, $B$, $D$ and
$N$ with $M\ge 2 B$, $B\mid M$, $D\le{M/B}$ and $N>M$,
$N$ items can be sorted with internal memory
size $M$, block size $B$ and $D$ disks with
\[
\frac{(3+g({m/D}))h(\log_m({{8 D}/B}))}{D}\cdot
\Tvn{Sort}_{M,B}(N)(1+F'(M,B,D,N))
\]
I/Os, where
$n=\Tceil{N/B}$, $m={M/B}$,
$\Tvn{Sort}_{M,B}(N)=2 n\Tceil{\log_m\! n}$,
\[
g(x)={\frac{\max\{5-2 x,0\}}{4 x-2}}\le{3/2}
\]
for $x\ge 1$,
$h(x)={1/{(1-({1/2})\max\{0,\min\{x,1\}\})}}\le 2$ for
arbitrary real $x$
and
\[
F'(M,B,D,N)=O\left(\frac{1}{\log_m\! n}
+\frac{D}{m\ln m}+\frac{\ln m}{\sqrt{D B}}
+\frac{1}{(D B)^{1/4}}+\frac{1}{\sqrt{D}}\right).
\]
\end{theorem}

\begin{proof}
Because of the term
${1/{\sqrt{D}}}$
in the
expression for $F'$, we can assume without
loss of generality that $m$ is larger than a
certain constant $m_0$,
which we choose as $m_0=2^{12}$.
As justified earlier, we will also
assume that $D\ge\sqrt{m}$.
We use the algorithm developed in the previous
subsections and
compute the parameter values as follows
($f$ is the function of
Lemma~\ref{lem:divisibility},
and $\widetilde{D}$, $r_2$ and $r_5$ are auxiliary quantities):

\begin{tabbing}
\quad\=\quad\=\quad\=\kill
\>$\displaystyle\DL:=\left\lceil
 \frac{D}{4(D B)^{1/4}}\right\rceil$;\\[4pt]
\>$\widetilde{D}:=\displaystyle\left\lfloor
 \frac{\min\{D,m-\DL\}}{2}\right\rfloor$;\\[4pt]
\>$(s,\overline{D}):=f\left(
 \max\left\{\left\lfloor\sqrt{{{\widetilde{D}}/B}}\right\rfloor,
 1\right\},
 \widetilde{D}\right)$;\\[4pt]
\>$D_5:=\displaystyle\min\left\{\left\lfloor\frac{m-\overline{D}-2\DL}{2}
 \right\rfloor,D\right\}$;\\[4pt]
\>$r_2:=
 \displaystyle\left\lfloor\frac{(m-1)s B}{\overline{D}-1}
 \right\rfloor-1$;\\[4pt]
\>$r_5:=\displaystyle\left\lfloor
 \frac{m-\overline{D}-D_5-2\DL}{s}\right\rfloor$;\\[4pt]
\>$r:=\min\{\Tfloor{{{r_2}/2}},r_5\}$;\\[4pt]
\>$D_2:=\displaystyle\min\left\{m-\left\lceil
 \frac{(r+1)(\overline{D}-1)}{s B}\right\rceil,
 D\right\}$;\\[4pt]
\>$D_4:=2\overline{D}$;
\end{tabbing}
Let $\DL$, etc., be the values computed above.
Because $m\ge D$ and $m\ge 2^8$,
$\DL\le\Tceil{{m/{(4 m^{1/4})}}}
\le\Tceil{{m/{16}}}\le{m/8}$.
Now $\widetilde{D}\ge
\min\{\Tfloor{{D/2}},\Tfloor{({7/{16}})m}\}$,
which is at least 2 because $D\ge 4$
and $m\ge 5$.
By the third property of Lemma~\ref{lem:divisibility},
we also have $\overline{D}\ge 2$.
Hence the computation of $r_2$ does
not lead to a division by zero and since
$s\ge 1$, all other
steps are easily seen to be well-defined.

We first argue that the constraints~(3) are satisfied.
All values computed are integers.
First, $1\le\overline{D}\le\widetilde{D}\le\Tfloor{D/2}
\le\Tceil{D/2}$.
Then, clearly, $D_2,D_4,D_5,\DL\le D$.
The relation ${{(r+1)(\overline{D}-1)}/{(s B)}}+D_2\le m$
holds because $D_2$
is computed precisely as
the largest integer solution bounded by~$D$
to this (linear) inequality.
Similarly,
$r_5 s+\overline{D}+D_5+2\DL\le m$
and
$r s+\overline{D}+D_5+2\DL\le m$ hold
because $r_5$ is computed as the largest
integer solution to the first inequality and $r\le r_5$.
Moreover,
$D_4+\DL=2\overline{D}+\DL
\le 2\widetilde{D}+\DL\le m$.
To see that $D_2\ge 1$, it suffices
to observe that
${{(r_2+1)(\overline{D}-1)}/{(s B)}}\le m-1$
and that $r\le r_2$.
We have $D_5\ge 1$ because
$m-\overline{D}-2\DL\ge{m/4}$ and
$m\ge 8$.
Since the first argument of $f$ is no larger
than its second argument, the first and last
properties of Lemma~\ref{lem:divisibility}
show that $s\mid\overline{D}$.
It is easy to see that $D_4\ge 1$,
$\overline{D}\mid D_4$
and $r\le m$.

What remains is to prove that
$r\ge 2$ or,
equivalently, that $r_2\ge 4$ and $r_5\ge 2$.
First,
\[
\displaystyle
\left\lfloor\sqrt{{\widetilde{D}}/{B}}\right\rfloor
\ge\left\lfloor\sqrt{\frac{\Tfloor{({7/{16}})D}}{B}}
 \right\rfloor
=\left\lfloor\sqrt{\frac{7 D}{16 B}}
 \right\rfloor
\]
and hence, by the second property of
Lemma~\ref{lem:divisibility},
$s\ge({1/2})\sqrt{{{7 D}/{(16 B)}}}$.
Now
${{(m-1)s B}/{(\overline{D}-1)}}
\ge 2 s B\ge\sqrt{{{7 D}/{16}}}\ge 5$,
where the last inequality follows from
$D\ge \sqrt{m}\ge 2^6$, and this shows that $r_2\ge 4$.

In order to demonstrate that $r_5\ge 2$, we prove that
$m-\overline{D}-D_5-2\DL\ge 2 s$.
Briefly let $u=m-\overline{D}-2\DL$.
Since $D_5\le\Tfloor{{u/2}}$,
$m-\overline{D}-D_5-2\DL=u-D_5\ge\Tceil{{u/2}}$.
But
\[
u\ge m-\left\lfloor\frac{m-1}{2}\right\rfloor
-\frac{m}{4}=m-1-\left\lfloor\frac{m-1}{2}\right\rfloor
+\frac{4-m}{4}
=\left\lceil\frac{m+2}{4}\right\rceil.
\]
It therefore suffices to show that
${{(m+2)}/8}>2 s-1$ or that $m>16 s-10$.
But since $s\le\sqrt{D}\le\sqrt{m}$,
this follows from $m\ge 2^8$.

We next bound the number of I/Os executed by
the algorithm, essentially by estimating the
values of its parameters.
First
\[
\displaystyle
\DL=\left\lceil\frac{D}{4(D B)^{1/4}}\right\rceil
=O\left(D\left(\frac{1}{(D B)^{1/4}}+\frac{1}{D}
\right)\right)
\]
and therefore
\[
\widetilde{D}=D\left(\frac{1}{2}-O\left(\frac{1}{(D B)^{1/4}}+\frac{1}{D}
\right)\right).
\]
By the third property of Lemma~\ref{lem:divisibility},
it follows that
\[
\overline{D}=D\left(\frac{1}{2}-O\left(\frac{1}{(D B)^{1/4}}
 +\frac{1}{\sqrt{D}}\right)\right).
\]
As observed earlier,
${{(m-1)s B}/{(\overline{D}-1)}}\ge r_2+1$.
But $r_2+1\ge 2 r+1\ge({3/2})(r+1)$ and
therefore
${{(r+1)(\overline{D}-1)}/{(s B)}}
 \le({2/3})(m-1)$.
As a consequence, $D_2=\Omega(D)$.
Because $D_4$ is just $2\overline{D}$,
\[
D_4=D\left(1-O\left(\frac{1}{(D B)^{1/4}}
 +\frac{1}{\sqrt{D}}\right)\right).
\]
Noting that
\[
2 g(x)+1=\max\left\{\frac{10-4 x+4 x-2}{4 x-2},1\right\}
=\max\left\{\frac{1}{{x/2}-{1/4}},1\right\}
\]
for $x\ge 1$, we find that
\begin{align*}
D_5&=\displaystyle\min\left\{\left\lfloor\frac{m-\overline{D}-2\DL}{2}
 \right\rfloor,D\right\}\\
&\ge D\left(\min\left\{\frac{m}{2 D}-\frac{1}{4},1\right\}-
O\left(
\frac{1}{(D B)^{1/4}}+\frac{1}{D}\right)\right)\\
&=D\left(\frac{1}{2 g({m/D})+1}-
O\left(\frac{1}{(D B)^{1/4}}+\frac{1}{D}\right)\right).
\end{align*}

\noindent
We are almost ready to sum the terms in~(2).
Observe first that since $D_4$, $D_5$ and $\overline{D}$
are all larger than $\epsilon D$ for some fixed $\epsilon>0$,
our bound of the form $D_4=D(1-O(\cdots))$
yields a bound of the form
${1/{D_4}}=({1/D})(1+O(\cdots))$, and analogously
for $D_5$ and $\overline{D}$.
We already noted that
$s B=\Omega(\sqrt{D B})$.
Hence
\[
\frac{1}{s B}\left(\frac{2}{D_2}+\frac{3}{\DL}\right)
=O\left(\frac{1}{\sqrt{D B}}\cdot\frac{1}{D}\left(
 1+{(D B)^{1/4}}\right)\right)
=O\left(\frac{1}{D}\cdot\frac{1}{(D B)^{1/4}}\right).
\]

\noindent
We also have
\[
\frac{z\ln m}{s D B}+\frac{n}{D}
=O\left(\frac{z}{D}\left(\frac{\ln m}{\sqrt{D B}}
 +\frac{1}{\log_m\! n}\right)\right).
\]
Now the total number of I/Os can be seen to be
\begin{align*}
z&\left(\frac{1}{D_4}+\frac{1}{D_5}+\frac{2}{\overline{D}}
 +\frac{1}{s B}\left(\frac{2}{D_2}+\frac{3}{\DL}\right)\right)
 +O\left(\frac{z\ln m}{s D B}+\frac{n}{D}\right)\\
=\frac{z}{D}&\left(6+2 g({m/D})+
O\left(
\frac{1}{\log_m\! n}
+\frac{\ln m}{\sqrt{D B}}
+\frac{1}{(D B)^{1/4}}+\frac{1}{\sqrt{D}}\right)\right).
\end{align*}
By the second property of Lemma~\ref{lem:divisibility},
$s\ge({1/2})\sqrt{{{\widetilde{D}}/B}}
\ge({1/2})\sqrt{{{\overline{D}}/B}}$
and therefore
${{m s B}/{\overline{D}}}
\ge({m/2})\sqrt{{B/{\overline{D}}}}
\ge m\sqrt{{B/{(2 D)}}}$
and
$\Tfloor{{r_2}/2}
\ge m\sqrt{{B/{(8 D)}}}-O(1)$.
Since $r_2=\Omega(\sqrt{m})$, this implies that
\[
\ln(\Tfloor{{r_2}/2})
\ge\ln m-({1/2})\ln({{8 D}/B})-O({1/{\sqrt{m}}}).
\]
Similarly,
$s\le\max\{\sqrt{{D/B}},1\}$ and therefore
\[
r_5\ge{\frac{\max\{m-2 D,{m/4}\}}{\max\{\sqrt{{D/B}},1\}}}-O(1).
\]
We essentially already observed that
$r_5\ge\Tfloor{{\sqrt{m}}/8}=\Omega(\sqrt{m})$, so
\begin{align*}
\ln r_5&\ge\ln(\max\{m-2 D,{m/4}\})-\max\{({1/2})\ln({D/B}),0\}
 -O({1/{\sqrt{m}}})\\
&=\ln m-\max\{({1/2})\ln({D/B}),0\}-O({D/m}).
\end{align*}
Since $r=\min\{\Tfloor{{r_2}/2},r_5\}$, we may conclude that
\begin{align*}
\ln r&\ge\ln m-\max\{({1/2})\ln({{8 D}/B}),0\}
 -O({D/m})\\
&=\ln m-({1/2})\max\{0,\min\{\ln({{8 D}/B}),\ln m\}\}
 -O({D/m}).
\end{align*}

\noindent
and
\begin{align*}
\frac{\ln m}{\ln r}
&\le\frac{1}{1-({1/2})\max\{0,\min\{\log_m({{8 D}/B}),1\}\}}
 +O\left(\frac{D}{m\ln m}
\right)\\
&={{h(\log_m({8 D}/B))}}
 +O\left(\frac{D}{m\ln m}
\right).
\end{align*}
Since
\[
z\le\Tvn{Sort}_{M,B}(N){\frac{\ln m}{2\ln r}}
 \left(1+\frac{1}{\log_m\! n}\right),
\]
the total number of I/Os can therefore also be bounded by
\[
\frac{(3+g({m/D}))h(\log_m({D/B}))}{D}
\cdot\Tvn{Sort}_{M,B}(N)\left(1+
F'(M,B,D,N)\right),
\]
where $F'$ is as in the theorem.
\end{proof}

\begin{remark}
No attempt was made to use
the smallest possible~$m_0$ in the
first part of the proof of Theorem~\ref{thm:main}.
The assumption $m\ge m_0=2^{12}$ can be
replaced by $m\ge 8$, $D\ge 4$
(but still also $D\ge\sqrt{m}$) and $B\ge 16$,
which excludes only cases of
scant practical interest.
To show this, we revise the first part of the
proof of Theorem~\ref{thm:main} and provide
new justification for the claims proved
using the large value of $m_0$.
Hence assume that $m\ge 8$, $D\ge 4$ and $B\ge 16$
and note that $B^{1/4}\ge 2$.

The first claim to reprove is that $\DL\le{m/8}$.
This holds for $m\le 16$ because
$D^{3/4}\le m^{3/4}\le 8$
and therefore $\DL\le\Tceil{({1/8})D^{3/4}}=1\le{m/8}$
and for $m\ge 16$ because
$m^{1/4}\ge 2$ and therefore
$\DL\le\Tceil{({1/8}){m/{m^{1/4}}}}
\le\Tceil{{m/{16}}}\le{m/8}$.

To show that $r_2\ge 4$, the proof of
Theorem~\ref{thm:main} argued that $2 s B\ge 5$,
a relation that is a triviality for $B\ge 16$.
In the case of $r_5\ge 2$, the relation
to show was $m>16 s-10$.
But for $8\le m\le 16$ this is clear
since $s=1$, while for $m\ge 16$
it follows from $s\le\sqrt{{D/B}}\le({1/4})\sqrt{m}$.

It is easy to see that the asymptotic assertions
of Theorem~\ref{thm:main} continue
to hold.
\end{remark}

If $m\ge({5/2})D$, $g({m/D})=0$,
and if $B\ge 8 D$, $h(\log_m({{8 D}/B}))=1$.
If both of the conditions
$m\ge({5/2})D$ and
$B\ge 8 D$ are satisfied, therefore,
the quantity
$(3+g({m/D}))h(\log_m({{8 D}/B}))$
of Theorem~\ref{thm:main} is~3.
On the other hand, it is never larger than
$(3+{3/2})\cdot 2=9$.

\section{Conclusion}

Guidesort is simple and based on a natural idea.
If a reasonably programmed algorithm for sorting
through multiway merging is available, the new
algorithm can be grafted onto it by replacing
the old subroutine for multiway merging by one
that executes Steps~1--5.
Step~1 is standard merging.
Because the sequences to be merged are smaller
by a factor of $B$ than the actual runs, in
many practical situations the merging can
be done any which way.
The algorithm proposed here
that can be viewed as the final part of sorting by
naive 2-way striping
is also very easy to implement.
Step~3 is an even simpler approximately
inverse operation.
Step~5 is standard merging, except that the tests
to identify the disks from which to input the next
blocks have been executed beforehand.
Step~2 implements a straightforward greedy
algorithm, and Step~4 is a trivial redistribution
of blocks
according to a precomputed pattern.

While most or all previous algorithms are
optimal with respect to internal computation,
i.e., execute $O(N\ln N)$ steps of internal
computation to sort $N$ items, it is not
obvious how to carry out the greedy coloring
of our Step~2 in less than
$\Theta({{\overline{D}}/s})$ time per block
(or $\Theta({{\overline{D}}/{(s w)}}+1)$ time per block
if $w$ bits can be manipulated together
in constant time),
for which reason we can bound the amount of
internal computation of \guidesort\ only by
$O(N\ln N+({{n D}/s})\log_m\! n)
=O(N(\ln N+\sqrt{{D/B}}\log_m\! n))$.
For most realistic parameter values, however---in
particular, if $D\le B$---this is still $O(N\ln N)$.

It may seem deplorable that the function~$F'$
should include an ``error term'' of
$O({1/{\log_m\! n}})$, which can be quite
large even for realistic values of the
parameters $n$ and~$m$.
However, recall that the term represents
just $O({n/D})$ I/Os, i.e., the cost of
a single pass over the input at full
parallel speed.
Approximating the total depth of the leaves
in the recursion tree
of a bottom-up or top-down algorithm
similar to those published to date
by $\Tceil{\log_m\! n}$
times the number of leaves incurs an error
of the same magnitude.
An ``error'' of $O({n/m})$, which is
comparable if $D\approx m$, seems even harder
to avoid:
The number of nodes in the recursion tree
is $\Theta({n/m})$, so if $\Omega(1)$ I/Os
are ``lost'' on average at every node, e.g., because of
a last block that is only partially
occupied by items, the accumulated
waste amounts to $\Omega({n/m})$ I/Os.

\bibliography{bpdm}

\end{document}